\documentclass[10pt, doublecolumn]{IEEEtran}
\usepackage{epsfig,latexsym}
\usepackage{float}
\usepackage{indentfirst}
\usepackage{amsmath}
\usepackage{amssymb}
\usepackage{times}
\usepackage{subfigure}
\usepackage{psfrag}
\usepackage{cite}
\usepackage{lastpage}
\usepackage{fancyhdr}
\usepackage{color}
 \usepackage{amsthm}
\usepackage{bigints}
\newtheorem{Lemma}{Lemma}
\newtheorem{Corollary}{Corollary}
\newtheorem{lemma}[Lemma]{$\mathbf{Lemma}$}

\newtheorem{corollary}[Corollary]{$\mathbf{Corollary}$}
\begin{document}%%
\title{ {\huge  Random Beamforming in Millimeter-Wave NOMA Networks}}

\author{ Zhiguo Ding, \IEEEmembership{Senior Member, IEEE}, Pingzhi Fan, \IEEEmembership{Fellow, IEEE}   and  H. Vincent Poor, \IEEEmembership{Fellow, IEEE}\thanks{
Parts of the material in this
paper were presented at IEEE International Conference in Communications, France, France.
Z. Ding and H. V. Poor   are with the Department of
Electrical Engineering, Princeton University, Princeton, NJ 08544,
USA.   Z. Ding is also with the School of
Computing and Communications, Lancaster
University, LA1 4WA, UK. P. Fan is with the Institute of Mobile
Communications, Southwest Jiaotong University, Chengdu, China.

 The work of Z. Ding was supported by the UK Engineering and Physical Science Research Council under grant number EP/L025272/1 and by the European Commission Program H2020-MSCA-RISE-2015 under grant number 690750. The work of P. Fan  was supported by the Huawei HIRP Flagship Project
(No.YB201504), the National Science and Technology Major Project
(No.2016ZX03001018-002) and the 111 project (No.111-2-14). The work of H. V. Poor was supported  by the U.S. National Science
Foundation under Grant CCF-1420575 and CNS-1456793.
}\vspace{-2em}} \maketitle
\begin{abstract}
This paper investigates the coexistence between two key enabling technologies for   fifth generation (5G) mobile networks, non-orthogonal multiple access (NOMA) and millimeter-wave (mmWave) communications. Particularly, the application of random beamforming to   mmWave-NOMA systems  is considered, in order to avoid the requirement that the base station know all the users' channel state information. Stochastic geometry is used  to characterize the performance of the proposed mmWave-NOMA transmission scheme, by using the key features of mmWave systems, e.g.,  mmWave transmission is highly directional and potential blockages will thin the user distribution. Two random beamforming approaches that can further reduce the system overhead  are also proposed, and their performance is   studied analytically in terms of    sum rates and   outage probabilities. Simulation results are also provided to demonstrate the performance of the proposed   schemes and verify the accuracy of the developed analytical results.
\end{abstract}\vspace{-1em}
\section{Introduction}
Non-orthogonal multiple access (NOMA) has recently received considerable  attention as a promising multiple access (MA) technique to be used in   fifth generation (5G) mobile  networks \cite{docom, metis}. Compared to conventional orthogonal multiple access (OMA), such as time division multiple access andor frequency division multiple access, NOMA encourages spectrum sharing among multiple users, rather than serving a single user in one orthogonal bandwidth block \cite{7021086, Nomading}. Sophisticated power allocation policies and detection methods, such as cognitive radio inspired power allocation,  superposition coding and  successive interference cancellation (SIC), are used to combat the co-channel interference which is not presented in OMA cases \cite{7263349, Zhiguo_CRconoma}. It is worth pointing out that the use of  NOMA can  still effectively support massive connectivity and efficiently meet   users' diverse QoS requirements, even if the users have similar channel conditions \cite{Zhiguo_iot}.

As an promising enabling technology  for 5G networks, NOMA has been shown to be compatible to many other 5G techniques, such as massive multiple-input multiple-output (MIMO), cognitive radio networks, as well as other types of MA techniques, such as orthogonal frequency division multiple  access (OFDMA) \cite{7398134,Zhiguo_massive, Robertnoma}. The purpose of this paper is to investigate the coexistence between NOMA and another important 5G technique, millimeter-wave (mmWave) communications \cite{6515173, 7434598, 7279196, 7110547,5621983}. Even though more bandwidth resources are available at very high frequencies,   the use of NOMA is still important for the following  reasons:
  \begin{itemize}
  \item The highly directional  feature of mmWave transmission implies that users' channels can be highly correlated, which potentially degrades the system performance. But such correlation is ideal for the application of NOMA.
  \item   The combination supports massive connectivity in  dense networks, e.g., where there are  hundreds of users to be connected in a small area.
       \item The rapid growth of mobile Internet services, particularly  emerging virtual reality (VR) and augmented reality (AR) services,  will dwarf the   radio spectrum gains obtained from the mmWave bands, which means that further improvement of the spectral efficiency is still important.

  \end{itemize}

In this paper, we consider a mmWave-NOMA downlink scenario, in which a base station equipped with multiple antennas communicates with multiple single-antenna nodes. While MIMO-NOMA has been extensively studied in \cite{7095538, 7383326,Zhiguo_mimoconoma}, the application of mmWave communications makes the addressed MIMO-NOMA scenario much different, mainly due to the characteristics   of mmWave propagation. The contributions of this paper are four-fold:

\begin{itemize}
\item We first  consider the application of random beamforming to the addressed mmWave-NOMA scenario, in which a single beam is randomly generated by the base station. While  random beamforming does not require the base station to know  all the users' channel vectors, conventional random beamforming still   requires all the users to send their scale channel gains to the base station, which can   consume significant  system overhead in a network with a large number of users. The fact that mmWave transmission  is highly directional is used in this paper to avoid scheduling those users who are likely to have low signal strength, which reduces the number of users who need to feed their channel quality information back to the base station and hence   reduces the system overhead. Stochastic geometry is applied to characterize the sum rate and the outage probabilities achieved by the proposed beamforming scheme, where the blockage feature of mmWave propagation  is also used to model the user distribution more realistically.

    \item In a fast time varying situation, in which the phases and the amplitudes of the users' channel gains change rapidly, a low-feedback transmission scheme is proposed by assuming that only the users' distance information is available to the base station. As a result, the users are ordered according to their path losses, instead of their effective channel gains. The impact of this partial channel state information (CSI) on the performance of the mmWave-NOMA downlink network is investigated.

\item A one-bit feedback random beamforming scheme is also proposed in order to further reduce the system overhead.   In particular, the base station sets a threshold which is broadcast to the users. Each user feeds one bit back to the base station to indicate its channel quality. The use of one-bit feedback can effectively reduce the amount of feedback, but will cause an ordering ambiguity at the base station. The impact of this ambiguity on the performance of the one-bit feedback transmission scheme is investigated. Furthermore, the effect of the threshold is also characterized, where the obtained analytical results show that a properly designed  threshold can ensure that the full diversity gain is achievable by the user selected to be the NOMA strong user.

\item The performance for the more challenging  scenario in which the base station generates multiple orthonormal beams is also investigated. Compared to the case with a single beam, each user in the scenario with multiple beams suffers more interference, including intra NOMA group interference and inter-beam interference. Because mmWave transmission is highly directional, inter-beam interference can be effectively suppressed by scheduling the users whose channel vectors are aligned with the randomly generated beams. Exact expressions for the outage probabilities achieved by the random beamforming scheme and their approximations are developed  in order to obtain greater insights.
\end{itemize}

\section{System Model}
Consider a mmWave-NOMA downlink transmission scenario with one base station   communicating  with multiple users, as shown in Fig. \ref{system_model}.   The base station is equipped with   $M$ antennas and each user has a single antenna.  Denote the disk   covered by the base station by $\mathcal{D}$. Assume that the base station is located at the origin of $\mathcal{D}$ and denote the   radius of the disk   by $R_{\mathcal{D}}$.    Assume that users are randomly deployed in the disc  following a homogeneous Poisson point process (HPPP) with density $\lambda$ \cite{Haenggi}.   Therefore, the number of users in the disk is Poisson distributed, i.e., $P(\text{K users in } \mathcal{D})=\frac{\mu^Ke^{-\mu}}{K!}$, where $\mu=\pi {R_{\mathcal{D}}}^2\lambda$.

\begin{figure}[!t]
\centering
\includegraphics[width=0.3\textwidth]{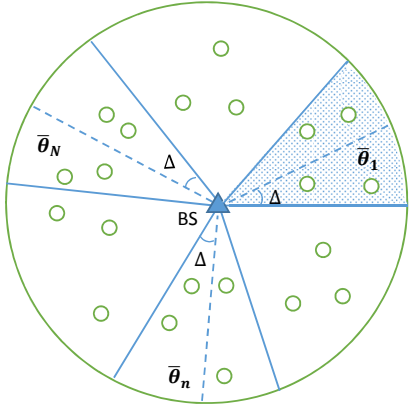}\vspace{-0.5em}
\caption{A system diagram for the addressed mmWave-NOMA scenario.  $\bar{\theta}_n$ denotes a beamforming vector randomly generated by the base station. Only the users that fall into a specific wedge-shaped sector will be scheduled, which is to ensure that the maximal angle difference between a scheduled user's channel vector and its associated beam is $\Delta$.  } \label{system_model}\vspace{-1em}
\end{figure}

%The path loss from the base station to user $k$ can be written as follows:
%\begin{align}
%L_k = \tau  + 10\alpha \log_{10} (d_k) + x_k,
%\end{align}
%where $\tau=20\log_{10}\left(\frac{4\pi}{\lambda_c}\right)$, $\lambda_c$ is the wavelength of the carrier signals, $\alpha$ is the path loss exponent, $d_k$ denotes the distance from the base station to user $k$, and $x_k$ is complex Gaussian distributed with zero mean and variance $\sigma^2_x$. In this paper, the users' distances will be treated as variables, as well as the parameter $x_k$, where  the constant $\tau$ will be omitted in the used channel model to ease the mathematical development.

As discussed in \cite{7434598} and \cite{7279196}, the mmWave channel model is quite different from those of   conventional lower frequency cellular networks; in particular,   the mmWave-based channel vector from the base station to user $k$ can be expressed as follows:
\begin{align}\label{channel model}
\mathbf{h}_k =  \sqrt{M} \frac{a_{k,0}\mathbf{a}(\theta^0_{k})}{\sqrt{ 1+d_k^{\alpha_{LOS}}}}+ \sqrt{M}\sum^{L}_{l=1}\frac{a_{k,l}\mathbf{a}(\theta^l_{k})}{\sqrt{1+ d_k^{\alpha_{NLOS}}}},
\end{align}
where $L$ is the number of multi-paths,  \begin{align}
\mathbf{a}(\theta) = \frac{1}{\sqrt{M}}\begin{bmatrix} 1 & e^{-j\pi \theta} &\cdots &e^{-j\pi(M-1) \theta} \end{bmatrix}^T,
\end{align}
$d_k$ denotes the distance between the transceivers, $\alpha_{NLOS}$ and $\alpha_{LOS}$ denote the path loss exponents for the non-line-of-sight (NLOS) and line-of-sight (LOS) paths, respectively,
$a_{k,l}$ denotes the complex gain for the $l$-th path and is complex Gaussian distributed, i.e.,  $a_{k,l}\sim CN(0,1)$, and $\theta_{k}^l$ denotes the normalized direction of the $l$-th path. We assume that the channel gains are independent from path to path. For notational simplicity,   the normalized direction of a path is treated the same as its physical angle of departure, and in Section \ref{subsection application}, we can show this simplification has no impact on the performance  of the proposed mmWave-NOMA scheme.

As discussed in  \cite{7279196} and \cite{6363891}, in mmWave communications,  the effect of LOS links is dominant, compared to those of NLOS links, e.g.,  the gain of an LOS link can be $20$ dB stronger than those of NLOS links. Therefore the first factor at the right-hand side of  \eqref{channel model} is dominant, which yields the following simplified channel model:
\begin{align}
\mathbf{h}_k =  \sqrt{M} \frac{a_{k}\mathbf{a}(\theta_{k})}{\sqrt{ 1+d_k^{\alpha}}},
\end{align}
where the subscripts $0$ and LOS have been omitted to simplify the notation.
%$y_k$ denotes the  log-normal fading gain, i.e., $y_k=10^{-\frac{x_k}{10}}$,, and

In practice, the direct path between the mmWave transceivers might be blocked by obstacles, which means that  an LOS path does not always exist. As a result, in addition to  path loss and fading attenuation,  mmWave transmission also suffers potential blockages,  which is an important feature to be captured. A simple way to model these blockages is to assume the existence of an LOS path if the distance between the transceivers  is smaller than a threshold \cite{7110547} and \cite{7360135}. Alternatively, a more sophisticated way   to model the probability of there being an LOS path for mm-Wave transmission has been introduced in \cite{5621983, 7448962} as follows:
\begin{align}
\label{blorckage}
\mathrm{P}(LOS) = e^{-\phi d_k},
\end{align}
where $\phi$ is determined by the building density, the shape of the buildings, etc. In this paper, we will use \eqref{blorckage} for modelling blockages in the addressed mmWave communication scenario. It is important to point out that these  blockages will thin  the node distribution, which will be discussed in detail in the next section.

\section{Random Beamforming: A Single-Beam Case}\label{section 2}
  Many existing precoding and beamforming schemes for NOMA require that the base station has access to the users' CSI. These approaches can consume a substantial  amount of system overhead, if there are many users in the system. In order to reduce the system overhead, we consider the application of random beamforming to mmWave-NOMA communication scenarios.
\subsection{The Application of Random Beamforming to NOMA}\label{subsection application}
In this section, we focus on the case in which a single beam, denoted by $\mathbf{p}$, is generated at the base station. Note that in the context of mmWave communications, analog  precoding is  preferable compared to digital precoding since the amplitude of a signal is kept constant and only the phase is changed. Therefore, following \cite{7279196} and \cite{7397837}, we use the following choice for beamforming:
 \begin{align}
 \mathbf{p} = \mathbf{a}(\bar{\theta}),
 \end{align}
 where $\bar{\theta}$ is uniformly distributed between $-1$ and $1$.  This choice of precoding is analog precoding since it alters the signal phase only, and keeps the signal modulus constant. It is worth pointing out that this beamformer  is also a special case of the  hybrid precoding design in \cite{7445130} with one radio frequency chain and $M$ antennas.

One straightforward solution for user scheduling is to ask each user to feed its effective channel gain $|\mathbf{h}_{j}^H   \mathbf{p}|^2$ back to the base station, and then the base station schedules the user with  the strongest channel. However, such an approach will still consume considerable  system overhead, particularly if there are many users in the cell.

In the context of mmWave communications, a useful observation is that many users do not have to participate in the competition for access to the channel, as explained in the following. Without loss of generality, user $j$ is randomly chosen to be served on beam $\mathbf{p}$. The effective channel gain of this user on the randomly generated beam, $|\mathbf{h}_{j}^H   \mathbf{p}|^2$, can be written as follows:
\begin{align}
|\mathbf{h}_{j}^H   \mathbf{p}|^2 &= M\frac{|a_{j}|^2|\mathbf{p}^H\mathbf{a}(\theta_{j})|^2}{{1+ d_j^{\alpha}}}=\frac{|a_{j}|^2\left| \sum^{M-1}_{l=0}e^{-j\pi l (\bar{\theta}-\theta_j)}  \right|^2}{{M(1+ d_j^{\alpha})}}.
\end{align}
Following   steps similar to those in \cite{7279196}, this effective channel gain can be rewritten as follows:
\begin{align}
|\mathbf{h}_{j}^H   \mathbf{p}|^2
&= \frac{|a_{j}|^2 \sin^2\left(\frac{\pi M (\bar{\theta} - \theta_j)}{2}\right) }{{M(1+ d_j^{\alpha})\sin^2\left(\frac{\pi  (\bar{\theta} - \theta_j)}{2}\right)}}
\\\nonumber &= \frac{|a_{j}|^2   }{{(1+ d_j^{\alpha})  }}F_M\left(\pi[\bar{\theta} - \theta_j]\right),
\end{align}
where $F_M(x)$ denotes the Fej\'er kernel. Note that   the Fej\'er kernel goes to zero quickly by increasing its argument, i.e.,  $F_M(x)\rightarrow 0$ for increasing  $x$. This means that a user can have a large effective channel gain on  beam $\mathbf{p}$   if this user's channel vector is aligned with the direction of the beam.

Following this observation, we will  schedule only the users who are located  in the wedge-shaped sector served by the beam, as highlighted   in Fig.~\ref{system_model}. Particularly, this   sector is denoted by $\mathcal{D}_{\theta}$, and its  central angle  is $2\Delta$, which means that the maximal angle difference between a scheduled user's channel vector and the beam is $\Delta$, and $\Delta\rightarrow 0$ is required to ensure a large effective channel gain. Note that,   when $\Delta\rightarrow 0$, the use of the normalized direction of a path to replace its  physical angle of departure has no impact on the performance of the proposed scheme, as illustrated in the following. Recall that the normalized direction $\theta$ is a function of the physical angle of departure, denoted by $\phi_\theta$, i.e., $\theta=\frac{2d \sin (\phi_\theta)}{\lambda}$, where $\lambda$ and $d$ are the carrier wavelength and the antenna separation distance, respectively. If $\Delta\rightarrow 0$, we have $|\bar{\theta}-\theta_j|\rightarrow 0$, and hence $|\phi_{\bar{\theta}} -\phi_{\theta_j}|\rightarrow 0$, which means that the two physical angles are very similar if the two normalized directions are similar.  Furthermore, as $\Delta\rightarrow 0$, the application of Taylor series leads to $\theta_j-\bar{\theta}\approx \frac{2d  \cos (\phi_{\bar{\theta}})}{\lambda} (\phi_{\theta_j}-\phi_{\bar{\theta}})$, and so our analytical results based on the normalized directions can be extended to the case with the physical angles   in a straightforward manner.

 \subsection{The Implementation of NOMA}
  Suppose  that there are $K$ users in the sector, $\mathcal{D}_{\theta}$, and these users are ordered according to their effective channel gains  as follows:
\begin{align}\label{order}
|\mathbf{h}_{1}^H   \mathbf{p}|^2 \leq \cdots \leq |\mathbf{h}_{K}^H   \mathbf{p}|^2.
\end{align}
Similarly to \cite{Zhiguo_CRconoma} and \cite{7095538}, we consider the case in which two users will be selected for the implementation of NOMA. Note that the implementation of NOMA in long term evolution advanced (LTE-A) is also based on the two-user case \cite{3gpp1}.  Since the aim of this paper is to study the impact of NOMA on mmWave communications, without loss of generality, we assume that user $i$  and user $j$  are paired together for NOMA transmission on a randomly generated beam. Note that $i$ and $j$ can be chosen arbitrarily, constrained by $1\leq i< j\leq K$. As a result, the performance of mmWave-NOMA with different scheduled users can be investigated, and the insights obtained from the performance analysis can offer guidelines for the  design of practical user scheduling algorithms. Therefore,   the signal sent by the base station is given by
 \begin{align}
 \mathbf{p} \left(\beta_{i} s_{i}+\beta_{j} s_{j}\right),
 \end{align}
 where $\beta_i$ denotes the power allocation coefficient. Since $|\mathbf{h}_{i}^H   \mathbf{p}|^2<|\mathbf{h}_{j}^H   \mathbf{p}|^2$, the application of NOMA means  $\beta_i\geq \beta_j$, where $\beta_{i}^2+\beta_{j}^2=1$.

  Therefore,  user $i$ will receive the following observation:
 \begin{align}
 y_{i} =& \mathbf{h}_{i}^H   \mathbf{p} \left(\beta_{i} s_{i}+\beta_{j} s_{j}\right)+n_{i},
 \end{align}
 where $n_{i}$ denotes additive Gaussian noise. User $i$ will treat its partner's message as noise and directly decode its   information  with the following signal-to-interference-plus-noise ratio (SINR):
 \begin{align}
  {\text{SINR}}_{i} =&\frac{|\mathbf{h}_{i}^H   \mathbf{p}|^2 \beta^2_{i} }{|\mathbf{h}_{i}^H   \mathbf{p}|^2 \beta^2_{ j}  +\frac{1}{\rho}} ,
 \end{align}
 where $\rho$ denotes the transmit signal-to-noise ratio (SNR).
As a result, the outage probability for user $i$ to decode its information is given by
\begin{align}
\mathrm{P}_{i|K}^o &= \mathrm{P}\left(\log(1+{\text{SINR}}_{i})<R_i|K\right)
= \mathrm{P}\left( {\text{SINR}}_{i}<\epsilon_i|K\right),
\end{align}
which is conditioned on the number of users in $\mathcal{D}_{\theta}$,
where $\epsilon_i=2^{R_i}-1$.

User $j$ first tries to decode its partner's message with the following SINR: $  {\text{SINR}}_{i\rightarrow j} =\frac{|\mathbf{h}_{j}^H   \mathbf{p}|^2 \beta^2_{i} }{|\mathbf{h}_{j}^H   \mathbf{p}|^2 \beta^2_{j} +\frac{1}{\rho}}$.
If ${\text{SINR}}_{i\rightarrow j} \geq \epsilon_{ i}$,   the user can decode its own message with the following SNR:
\begin{align}
 \text{SINR}_{j} =&\rho |\mathbf{h}_{j}^H   \mathbf{p}|^2 \beta^2_{j} ,
 \end{align}
 after removing its partner's information, a procedure known as SIC.
 Therefore the outage probability experienced by user $j$ can be expressed as follows:
 \begin{align}\label{overall outage}
\mathrm{P}_{j|K}^o =  1 - \mathrm{P}\left( {\text{SINR}}_{i\rightarrow j} >\epsilon_i,\text{SINR}_{j}>\epsilon_j |K\right),
\end{align}
which is again conditioned on $K$.

 As a result, the outage sum rate achieved by the mmWave-NOMA transmission scheme can be expressed as follows:
 \begin{align}
 R^{NOMA}_{sum} &= \mathrm{P}(K=1) (1-\mathrm{P}_{OMA}^{1|K})R_{1}+ \sum^{\infty}_{k=2}\mathrm{P}(K=k) \nonumber  \\ \label{sum rate} &\times  \left((1-\mathrm{P}_{i|K}^o)R_i + (1-\mathrm{P}_{j|K}^o)R_j\right),
 \end{align}
 and the sum rate achieved by mmWave-OMA can be expressed similarly as follows:
  \begin{align}
 R^{OMA}_{sum} &= \mathrm{P}(K=1) (1-\mathrm{P}_{OMA}^{1|K})R_{1}+ \sum^{\infty}_{k=2}\mathrm{P}(K=k) \nonumber \\\label{sum rate1} &\times\left((1-\mathrm{P}_{OMA}^{i|K})R_i + (1-\mathrm{P}_{OMA}^{j|K})R_j\right),
 \end{align}
 where $\mathrm{P}_{OMA}^{n|K}$ denotes the conditional outage probability when OMA is used. The reason for using the OMA mode in \eqref{sum rate} is that it is possible to have a single   user in $\mathcal{D}_{\theta}$. In this case,  NOMA cannot be implemented and we simply use OMA, i.e., $\mathrm{P}_{OMA}^{n|K}=
 \mathrm{P}\left(\log(1+\rho |\mathbf{h}_{n}^H   \mathbf{p}|^2)  <2R_{n} \right)$, for $n \in\{i,j\}$\footnote{One can also use $\mathrm{P}_{OMA}^{1|K}=
 \mathrm{P}\left(\log(1+\rho |\mathbf{h}_{1}^H   \mathbf{p}|^2)  <R_{1} \right)$ for the case $K=1$, which will make the notation in \eqref{sum  rate} and \eqref{sum  rate1} more complicated. It is worth pointing out that the probability of having  $K=1$ is very small and different designs for this trivial case do not cause much difference to the overall sum rate. }.

%   Given different realization of $K$, the use of a fixed choice of $i$ and $j$ results in situations, where the preferred choices of $i$ or $j$ are not possible. Following \cite{Zhiguo_CRconoma}, a preferable choice for user pairing is to pair two users whose channel conditions are very different. For this reason, we focus on the case of $i=1$ and $j=K$, i.e., whenever user $1$ is scheduled to transmit, we pair it with user $K$ by using the NOMA principle.

\subsection{Characterization of the Sum Rate and Outage Probabilities}
In order to evaluate the sum rate shown in \eqref{sum rate}, it is important to find expressions for the outage probabilities, $\mathrm{P}_{j|K}^o$ and $\mathrm{P}_{i|K}^o$; these  are   related to the probability density function (pdf) of the ordered channel gain, $|\mathbf{h}_{j}^H   \mathbf{p}|^2$, which is provided in the following lemma.
\begin{lemma}
Suppose that there are $K$ users in $\mathcal{D}_{\theta}$. The pdf of the ordered channel gain, $|\mathbf{h}_{j}^H   \mathbf{p}|^2$, is given by
\begin{align}
f_{|\mathbf{h}_{j}^H   \mathbf{p}|^2}(z) = c_j \frac{dF_{\pi(j)}(z)}{dz}F^{j-1}_{\pi(j)}(z)\left(1-F_{\pi(j)}(z)\right)^{K-j},
\end{align}
where  $c_j=\frac{K!}{(j-1)!(K-j)!}$,
\begin{align}
F_{\pi(j)}(y)=    \int^{\bar{\theta}+\Delta}_{\bar{\theta}-\Delta}\int^{R_{\mathcal{D}}}_{0}\left(1- e^{-
   \frac{{y(1+ r^{\alpha})  }}{  F_M\left(\pi[\bar{\theta} - \theta]\right) }
  }\right)\\\nonumber \times  \frac{\lambda\phi^{2} e^{-\phi r}}{2\Delta \lambda\gamma(2, R_{\mathcal{D}}\phi )}rdrd\theta,
\end{align} and  $\gamma(\cdot)$ denotes the incomplete gamma function.
\end{lemma}
\begin{proof}
The density function in the lemma can be evaluated  by first characterizing the unordered channel gains and then applying the theory of order  statistics.

First we focus on an unordered channel gain, denoted by $|\mathbf{h}_{\pi(j)}^H   \mathbf{p}|^2$. Denote the location of this node by $x_{\pi(j)}$, where its probability distribution and pdf are denoted by  $P_{X_{\pi(j)}}$ and $p_{X_{\pi(j)}}$, respectively.  In this case we can find the cumulative distribution function (CDF) of the unordered channel gain as follows:
\begin{align}\nonumber
F_{{\pi(j)}}(y) &=  \underset{\mathcal{D}_{\theta}}{\int}\mathrm{P}\left(|\mathbf{h}_{\pi(j)}^H   \mathbf{p}|^2<y~|~ X_{\pi(j)}=x_{\pi(j)}\right)dP_{X_{\pi(j)}}\\\nonumber &=  \underset{\mathcal{D}_{\theta}}{\int}\left(1- e^{-
   \frac{{y(1+ r(x)^{\alpha})  }}{  F_M\left(\pi[\bar{\theta} - \theta_{\pi(j)}]\right) }
  }\right)p_{X_{\pi(j)}}(x)dx,
\end{align}
where $r(x)$ denotes the distance from the origin to   point $x$. Note that the conditioning on  $K$ has been omitted since it does not affect the CDF.

It is important to note that the nodes participating in NOMA  no longer follow the original HPPP with parameter $\lambda$, because of potential blockages. Particularly, with the blockage model in \eqref{blorckage}, it is less likely for  a user far away from the base station to have an LOS path. Therefore,  following the discussions in \cite{Wangpoor11}, the effect of blockages is to thin the original homogeneous point process and this thinning process yields another PPP with the following intensity:
\begin{align}
\lambda_{\Phi_2}(x) = \lambda e^{-\phi r(x)}.
\end{align}
Therefore,  the mean measure for this new PPP, denoted by $\mu_{\Phi_2}(\mathcal{D}_\theta)$, can be obtained as follows:
\begin{align}\label{mu}
\mu_{\Phi_2}(\mathcal{D}_\theta) &=\underset{\mathcal{D}_\theta}{\int} \lambda_{\Phi_2}(x) dx\\\nonumber &
=\int^{\bar{\theta}+\Delta}_{\bar{\theta}-\Delta}\int^{R_{\mathcal{D}}}_{0}\lambda e^{-\phi r}r dr d\theta = 2\Delta \lambda \phi^{-2}\gamma(2, R_{\mathcal{D}}\phi ).
\end{align}

As a result, after considering potential blockages, the probability of having  $K$ users in the sector, $\mathcal{D}_{\theta}$, can be obtained as follows:
\begin{align}\label{k1}
\mathrm{P}(K=k) = \frac{\left(\mu_{\Phi_2}(\mathcal{D}_\theta)\right)^k}{k!}e^{-\mu_{\Phi_2}(\mathcal{D}_\theta)}.
\end{align}
Since the intensity and the mean measure of the new PPP are known,   the pdf of $x_{\pi(j)}$ can be written as follows:
\begin{align}
p_{X_{\pi(j)}}(x) &=\frac{\lambda_{\Phi_2}(x)}{\mu_{\Phi_2}(\mathcal{D}_\theta) } =\frac{\lambda\phi^{2} e^{-\phi r(x)}}{2\Delta \lambda\gamma(2, R_{\mathcal{D}}\phi )}.
\end{align}
Accordingly, the CDF of the unordered channel gain can be written as follows:
\begin{align}\nonumber
F_{\pi(j)}(y)=  \underset{\mathcal{D}_{\theta}}{\int}\left(1- e^{-
   \frac{{y(1+ r(x)^{\alpha})  }}{  F_M\left(\pi[\bar{\theta} - \theta_{\pi(j)}]\right) }
  }\right)\frac{\lambda\phi^{2} e^{-\phi r(x)}}{2\Delta \lambda\gamma(2, R_{\mathcal{D}}\phi )}dx,
\end{align}
and  by using polar coordinates, the expression for $F_{\pi(j)}(y)$ in the lemma can be obtained.
After using the assumption  that all the channel gains are independent and identically distributed and also  applying the theory of order  statistics \cite{David03}, the proof is complete.
\end{proof}
By applying the above lemma and also  some algebraic manipulations, $\mathrm{P}_{j|K}^o$ and $\mathrm{P}_{i|K}^o$ can be obtained   in the following corollary.
\begin{corollary}
By using the proposed mmWave-NOMA transmission scheme, the outage probability experienced by user $j$ conditioned on $K$ is given by
 \begin{align}\label{1 outage}
\mathrm{P}_{j|K}^o   &=   c_j\sum^{K-j}_{p=0}{K-j \choose p} (-1)^p \frac{F^{j+p}_{\pi(j)}(\eta_j)}{j+p},
\end{align}
if $\beta_i^2>\beta_j^2\epsilon_i$, otherwise $\mathrm{P}_{j|K}^o=1$,  where $\eta_j = \max\left\{
\frac{\frac{\epsilon_i}{\rho}}{\beta_i^2-\beta_j^2\epsilon_i}, \frac{\epsilon_j}{\rho \beta_j^2}\right\}$. The conditional  outage probability for user $i$ is given by
 \begin{align}\label{2 outage}
\mathrm{P}_{i|K}^o &=    c_i\sum^{K-i}_{p=0}{K-i \choose p} (-1)^p \frac{F^{i+p}_{\pi(j)}(\eta_i)}{i+p},
\end{align}
where $\eta_i =
\frac{\frac{\epsilon_i}{\rho}}{\beta_i^2-\beta_j^2\epsilon_i} $.
\end{corollary}
By using the above corollary  and substituting \eqref{k1}, \eqref{1 outage} and \eqref{2 outage} into \eqref{sum rate} and \eqref{sum rate1}, the sum rates achieved by mmWave-NOMA and mmWave-OMA can be calculated.
\subsection{Asymptotic Performance Analysis}
The obtained results shown in \eqref{1 outage} and \eqref{2 outage} are quite complicated, since they involve the calculation of double integrals. In order to obtain some insight, we will obtain    approximations to these expressions. Particularly, our asymptotic studies are carried out by using the following  two assumptions. One is that the central angle of the sector, $2\Delta$, is small, i.e., $\Delta \rightarrow 0$, and the other is the high SNR assumption. The use of these two assumptions leads to the following lemma.
 \begin{lemma}
 When $\Delta\rightarrow 0$ and at high SNR, the conditional outage probabilities $\mathrm{P}_{i|K}^o$ and $\mathrm{P}_{j|K}^o$ can be approximated as follows:
  \begin{align}
\mathrm{P}_{k|K}^o   &\approx    c_k   \frac{F^{k+p}_{\pi(j)}(\eta_k)}{k} ,
\end{align}
where $k\in\{i,j\}$. The diversity gain available at user $k$ is $k$.
 \end{lemma}
 \begin{proof}
In order to use the assumption $\Delta \rightarrow 0$,  recall that the Fej\'er kernel can  be  written as follows:
\begin{align}
  F_M\left(\pi[\bar{\theta} - \theta]\right)=
\frac{ \sin^2\left(\frac{\pi M (\bar{\theta} - \theta)}{2}\right) }{{M \sin^2\left(\frac{\pi  (\bar{\theta} - \theta)}{2}\right)}}.
\end{align}
Note that  $|\bar{\theta} - \theta|\leq\Delta$.  When $\Delta$ is small, the Fej\'er kernel can be approximated as follows:
\begin{align}\label{approximtion x}
F_M\left(\pi[\bar{\theta} - \theta]\right)
&\approx M  \text{sinc}^2\left(\frac{\pi M (\bar{\theta} - \theta)}{2}\right) \\\nonumber &\approx  M  \left(1-\frac{\pi^2 M^2 (\bar{\theta} - \theta)^2}{12}\right) , %&\approx M  \left(1+\sum^{N_\theta}_{n=1}(-1)^n\frac{\pi^{2n} M^{2n} (\bar{\theta} - \theta_k)^{2n}}{2^{2n}(2n+1)!}\right).
\end{align}
where   the first approximation  follows from $\sin(x)\approx x$ for $x\rightarrow 0$, and the second approximation is due to the two following facts: $\text{sinc}(x)\approx 1-\frac{x^2}{6}$ and $(1-x)^2\approx 1-2x$, for $x\rightarrow   0$ \cite{Gaussiandd}.

Therefore, the CDF of an unordered channel gain can be approximated as follows:
\begin{align}
F_{{\pi(j)}}(y)&\approx   \int^{\bar{\theta}+\Delta}_{\bar{\theta}-\Delta}\int^{R_{\mathcal{D}}}_{0}\frac{\lambda\phi^{2} e^{-\phi r}}{2\Delta \lambda\gamma(2, R_{\mathcal{D}}\phi )} \\\nonumber &\times \left(1- e^{-
   \frac{{y(1+ r^{\alpha})  }}{   M  \left(1-\frac{\pi^2 M^2 (\bar{\theta} - \theta)^2}{12}\right)  }
  }\right)rdrd\theta
  \\ \nonumber
  &\approx   \int^{\bar{\theta}+\Delta}_{\bar{\theta}-\Delta}\int^{R_{\mathcal{D}}}_{0}\frac{\lambda\phi^{2} e^{-\phi r}}{2\Delta \lambda\gamma(2, R_{\mathcal{D}}\phi )} \\\nonumber &\times  \left(1- e^{-
   \frac{{y(1+ r^{\alpha})  }}{   M    }\left(1+\frac{\pi^2 M^2 (\bar{\theta} - \theta)^2}{12}\right)
  }\right)rdrd\theta,
\end{align}
where the last approximation follows from $(1-x)^{-1}\approx 1+x$, for $x\rightarrow   0$.

After applying the assumption that $\Delta\rightarrow  0$, we will further apply the high SNR approximation. Note that at high SNR, both $\eta_i$ and $\eta_j$ go to zero, which means
\begin{align}
F_{{\pi(j)}}(\eta_i)
  &\approx   \int^{\bar{\theta}+\Delta}_{\bar{\theta}-\Delta}\int^{R_{\mathcal{D}}}_{0}\frac{\lambda\phi^{2} e^{-\phi r}}{2\Delta \lambda\gamma(2, R_{\mathcal{D}}\phi )}  \left(
   \frac{{\eta_i(1+ r^{\alpha})  }}{   M    }\right. \\\nonumber &\times \left.\left(1+\frac{\pi^2 M^2 (\bar{\theta} - \theta)^2}{12}\right)
  \right)rdrd\theta.
\end{align}
After  some algebraic manipulations, the CDF for an unordered channel gain can be approximated as follows:
\begin{align}
F_{{\pi(j)}}(\eta_i)
  &\approx \frac{\eta_i  }{2M   \gamma(2, R_{\mathcal{D}}\phi )}  \left(2+\frac{\pi^2 M^2 \Delta^2}{18}\right)  \\\nonumber &\times  \left( \gamma(2,R_{\mathcal{D}}\phi)+ \phi^{-\alpha}\gamma(\alpha+2,R_{\mathcal{D}}\phi)\right)\dot\sim \frac{1}{\rho},
  \end{align}
  where $f(\rho)\dot\sim \frac{1}{\rho^x}$ when $\underset{\rho\rightarrow \infty}{\lim}\frac{\log (f(\rho))}{\log \rho}=-x$ \cite{Zhengl03}.

By using the above approximations, the outage probability at user $i$   can be approximated as follows:
 \begin{align}
\mathrm{P}_{i|K}^o &=    c_i\sum^{K-i}_{p=0}{K-i \choose p} (-1)^p \frac{F^{i+p}_{\pi(j)}(\eta_i)}{i+p}
 \\\nonumber &\approx    c_i   \frac{F^{i+p}_{\pi(j)}(\eta_i)}{i}\dot\sim \frac{1}{\rho^{i}},
\end{align}
which means that the diversity gain at user $i$ is $i$. The results for user $j$ can be obtained similarly,  and the proof is complete.
\end{proof}
{\it Remark:} Note that an implication of having  a small $\Delta$ is that the area of the sector becomes so small that there might be no user in it. But in many practical scenarios, such as in a sport stadium or a conference hall, the users are so densely deployed that it is always possible to find multiple users located in a sector even with  a small $\Delta$.
\section{Random Beamforming with Limited Feedback}
In the previous section, it is assumed that the base station has  perfect knowledge of the users' effective channel gains. However, for a fast time varying situation, this assumption might not be realistic, since  the phases of the channel vectors and their fading coefficients, $\theta_k$ and $a_k$, are changing rapidly. In this section, we investigate two   random beamforming  schemes with low system overhead.

 \subsection{With the Distance Information Available at the Base Station}\label{subsection ditance} Compared to the phases and fading coefficients of the channels, the users' distance information will change  relatively slowly, which means that it is more realistic for the base station   to have access to the users' distance information only. Therefore, in this subsection, we investigate the impact of this partial CSI on the performance of mmWave-NOMA.

Again assume  that only the users that fall into the sector $\mathcal{D}_{\theta}$ will participate in the NOMA transmission. Assume that there are $K$ users in this sector.
Since the users' distances are known,  the base station will order the users according to the following criterion:
\begin{align}
d_{ 1}\leq \cdots \leq d_{ K},
\end{align}
instead of using the effective channel gains which are not known to the base station. Similarly to the previous section, we schedule user $i$ and user $j$ for the NOMA transmission to act as the weak and strong users, respectively.   Since a user with  a shorter distance has a stronger channel condition, we take  $i>j$.

Note that the density functions of the ordered distances have been found in \cite{1512427} when the users are distributed randomly in a ball.  The shape of the addressed area is a sector, but the steps provided in  \cite{1512427} are still applicable, as shown in the following. Particularly,  the CDF of $d_k$ can be calculated from the probability of the event that there are less than $k$ users inside a sector with radius $r$, i.e.,
\begin{align}
F_{d_k}(r) &= 1 - \sum^{k-1}_{i=0}\mathrm{P}(E_i)\\\nonumber &=1 - \sum^{k-1}_{i=0}e^{- \mu_{\Phi_2}(\mathcal{A}(r))}\frac{\left( \mu_{\Phi_2}(\mathcal{A}(r))^i
\right)}{i!},
\end{align}
where  $\mathcal{A}(r)$ denotes  a sector with radius of $r$, and $E_i$ denotes the event that there are $i$ users in  $\mathcal{A}(r)$.
Following steps similar to those for obtaining \eqref{mu}, the factor $\mu_{\Phi_2}(\mathcal{A}(r))$ can be found as follows:
\begin{align}
\mu_{\Phi_2}(\mathcal{A}(r))=2\Delta \lambda \phi^{-2}\gamma(2, r\phi ).
\end{align}
Substituting the expression for  $\mu_{\Phi_2}(\mathcal{A}(r))$  into the CDF expression, the CDF of $d_k$ can be expressed as follows:
\begin{align}
F_{d_k}(r)
&=1 - \sum^{k-1}_{i=0}e^{-2 \Delta \lambda \phi^{-2}\gamma(2, r\phi )}  \frac{\left(2 \Delta \lambda \phi^{-2}\gamma(2, r\phi )\right)^i}{i!}.
\end{align}
As a result, the corresponding pdf for the $k$-th smallest distance can be found as follows:
\begin{align}\label{pdf ppp}
f_{d_k}(r)
&=   2 \Delta \lambda e^{-r\phi} r e^{-2 \Delta \lambda \phi^{-2}\gamma(2, r\phi )}    \frac{\left(2 \Delta \lambda \phi^{-2}\gamma(2, r\phi )\right)^{k-1}}{(k-1)!},
\end{align}
where we have used the fact that
\[
\frac{d\gamma(2, r\phi )}{dr} = e^{-r\phi} r\phi^2.
\]
The difference between the above pdf expression and the one in \cite{1512427} is due to the   facts that the area for the addressed problem is not a ball and the addressed density is   a function of $r$.

On the other hand, note that the angle of user $k$'s channel vector is independent of its distance, and it is uniformly distributed between $(\bar{\theta}-\Delta)$ and $(\bar{\theta}+\Delta)$.
Therefore the CDF of user $k$'s channel gain can be obtained as follows:
\begin{align}\label{fx}
F_{k}(y)
  &=  \underset{\mathcal{D}_{\theta}}{\int}\left(1- e^{-
   \frac{{y(1+ r(x)^{\alpha})  }}{  F_M\left(\pi[\bar{\theta} - \theta_{\pi(j)}]\right) }
  }\right)p_{X_{\pi(j)}}(x)dx\\\nonumber &=
 \int^{\bar{\theta}+\Delta}_{\bar{\theta}-\Delta}\int^{R_{\mathcal{D}}}_0\left(1- e^{-
   \frac{{y(1+ r^{\alpha})  }}{  F_M\left(\pi[\bar{\theta} - \theta]\right) }
  }\right)\frac{f_{d_k}(r)}{2\Delta} drd\theta.
\end{align}

It is important to point out that the above CDF is valid only if we can find the $k$-th nearest node. Or in other words, if there is no boundary to $\mathcal{D}_{\theta}$ and the nodes are spread  throughout of the plane, the above CDF can be applied. For the addressed scenario, the users are confined in $\mathcal{D}_{\theta}$, i.e., $r\leq R_{\mathcal{D}}$, which means that it is possible that the $k$-th nearest node does not exist, i.e., there are fewer than $(k-1)$ nodes in $\mathcal{D}_{\theta}$. By using the result in \eqref{fx} and also considering the possible choices  for the number of users in $\mathcal{D}_{\theta}$,  we can obtain  the following lemma for the outage probability  and the sum rate.

 \begin{lemma} \label{lemmax1}
When only the users' distance information is available,  the outage probability for the $k$-th nearest node can  be written    as follows:
\begin{align}\label{Fxx}
F^o_{k} = \sum^{k-1}_{n=0}\mathrm{P}(K=n) + \left(1-\sum^{k-1}_{n=0}\mathrm{P}(K=n)\right) F_{k}(\eta_k),
\end{align}
where $k\in\{i,j\}$. Moreover, the outage sum rate can be shown as follows:
 \begin{align}\label{sum distance}
 R^{NOMA}_{sum} &=  (1-F_{j}(\eta_j))R_j+ (1-F_{i}(\eta_i))R_i,
 \end{align}
where the $k$-th nearest user has a targeted data rate of $R_k$.
\end{lemma}
{\it Remark 1:} It is important to point out that the sum rate in \eqref{sum distance} means that no transmission will take place if  the $i$-th nearest user cannot be found in $\mathcal{D}_{\theta}$, and the NOMA transmission is adopted even if the $j$-th nearest user can be found but the $i$-th one cannot. Note that other transmission strategies can also be used for these trivial  cases which happen with low probabilities in a densely deployed  network.

{\it Remark 2:} Note that one can also use a CDF expression conditioned on $K$ to find the outage probability, but this is difficult to evaluate since the conditioning on  $K$ converts the Poisson point process to a Bernoulli one to which the result in \eqref{pdf ppp} is not applicable.

\subsubsection*{Asymptotic performance analysis}
While the  expressions for the outage probability and the sum rate in Lemma \ref{lemmax1} can be   calculated numerically, approximations are still desirable  in order to obtain greater insight. Following steps similar to those in the previous section, i.e.,  when $\Delta$ approaches zero, the Fej\'er kernel can be simplified,  which yields the following approximation:
\begin{align}
F_{k}(y)  &\approx
 \int^{\bar{\theta}+\Delta}_{\bar{\theta}-\Delta}\int^{R_{\mathcal{D}}}_0\left(1- e^{-
   \frac{{y(1+ r^{\alpha})  }}{M  \left(1-\frac{\pi^2 M^2 (\bar{\theta} - \theta)^2}{12}\right) }
  }\right)  \\\nonumber &\times  f_{d_k}(r) dr\frac{1}{2\Delta}d\theta.
\end{align}
Furthermore  notice that both $\eta_i$ and $\eta_j$ approach zero at high SNR, which yields the following approximation:
\begin{align}\nonumber
F_{i}(\eta_i)  \approx&
 \int^{\bar{\theta}+\Delta}_{\bar{\theta}-\Delta}\int^{R_{\mathcal{D}}}_0\left(
   \frac{{\eta_i(1+ r^{\alpha})  }}{M  \left(1-\frac{\pi^2 M^2 (\bar{\theta} - \theta)^2}{12}\right) }
  \right)\frac{ f_{d_i}(r)}{2\Delta} drd\theta\\\nonumber \approx&
 \int^{R_{\mathcal{D}}}_0\left(
   \frac{{\eta_i(1+ r^{\alpha})  }}{M   }
  \right)f_{d_i}(r) dr  \frac{ 1}{2\Delta} \\ &\times \int^{\bar{\theta}+\Delta}_{\bar{\theta}-\Delta}\left(1+\frac{\pi^2 M^2 (\bar{\theta} - \theta)^2}{12}\right)d\theta,
\end{align}
and $F_{j}(\eta_j)$ can be obtained similarly.   The integral over $\theta$ can be obtained  by following steps similar to those in the previous section. In addition, define the integral  over  $r$, a factor not related to the transmit SNR, as follows:
\begin{align}
Q_{1,i}\triangleq &\int^{R_{\mathcal{D}}}_0\left(
   \frac{{ (1+ r^{\alpha})  }}{M   }
  \right)  \lambda e^{-r\phi} r e^{-2 \Delta \lambda \phi^{-2}\gamma(2, r\phi )} \\\nonumber &\times    \frac{\left(2 \Delta \lambda \phi^{-2}\gamma(2, r\phi )\right)^{i-1}}{(i-1)!} dr.
\end{align}

Therefore the outage probability can be approximated as follows:
\begin{align}\nonumber
F_{i}(\eta_i)   &\approx   \left(1+\frac{\pi^2 M^2 \Delta^2}{36}\right)Q_{1,i}\eta_i ,
\end{align}
which demonstrates that the use of distance information only yields a  diversity gain of one for all the users. This is expected since the base station has access to partial CSI only and the dynamics of  the   fading  gains cannot be used.

\subsection{With One-Bit Feedback}\label{subsection 2}
In the case in which the number of users in the sector $\mathcal{D}_{\theta}$ is very large, feeding these users' effective channel gains or distances back to the base station can still be very demanding. As an alternative, asking each user to feed only one bit about its channel quality  back to the base station can substantially  reduce the system overhead.

In particular,  the base station will first set a threshold, $\xi$, $\xi>0$,  which will be broadcast to all the users prior to the downlink transmission.
Each user in the sector will compare its effective channel gain with  $\xi $ and send $1$ to the base station if its channel gain is larger than $\xi$, otherwise it will send $0$ to the base station. As a result, the users in the sector will be divided into two sets, denoted by $\mathcal{S}_1$ and $\mathcal{S}_2$, respectively. Particularly, the users in $\mathcal{S}_2$ are the ones which feed $1$ back to the base station, i.e.,
\begin{align}
\mathcal{S}_2 \triangleq \{i| x_i\in \mathcal{D}_{\theta}, |\mathbf{h}_{i}^H   \mathbf{p}|^2>\xi\},
\end{align}
and $\mathcal{S}_1$ is defined similarly by grouping those users whose feedbacks are $0$.

When there is more than one user in $\mathcal{D}_{\theta}$, i.e., $K\geq 2$, and $|\mathcal{S}_n|\neq 0$, $n\in\{1,2\}$, the base station will randomly select one user from $\mathcal{S}_1$ to be paired with another user randomly selected from $\mathcal{S}_2$. If all the $K$ nodes are in one group, the base station will randomly select two users from this group for the implementation of NOMA. If there is only one user in the sector, i.e., $K=1$, this user will be served solely by the base station.   No user will be served if both sets are empty, which happens only if $K=0$. In the following, we will focus on the case with $K\geq 2$.

The following lemma provides the outage probabilities for the users selected to act as the NOMA strong and weak users, respectively.
\begin{lemma}
Suppose that there are  $K\geq 2$ users in $\mathcal{D}_{\theta}$. When each user only feeds one bit back to the base station using the above protocol, the outage probability for the user selected to act as the weak user is given by
\begin{align}
\mathrm{P}_{\mathcal{S}_1}^o =&
  F_{\mathcal{S}_1|K}(\tilde{\eta}_1)
 \sum^{K}_{n=1}\mathrm{P}(|\mathcal{S}_1|=n)  +\mathrm{P}(|\mathcal{S}_1|=0)   F_{\mathcal{S}_2|K}\left(\tilde{\eta}_1\right),
\end{align}
and the outage probability for the user selected to act as the strong user is given by
\begin{align}
\mathrm{P}_{\mathcal{S}_2}^o =&
  F_{\mathcal{S}_2|K}(\tilde{\eta}_2)
 \sum^{K}_{n=1}\mathrm{P}(|\mathcal{S}_2|=n)  +\mathrm{P}(|\mathcal{S}_2|=0)    F_{\mathcal{S}_1|K}\left(\tilde{\eta}_2\right),
\end{align}
where
$\mathrm{P}(|\mathcal{S}_2|=n)  = {K \choose n} \left(F_{{\pi(j)}}(\xi)\right)^{K-n} \left(1-F_{{\pi(j)}}(\xi)\right)^n$, $
\mathrm{P}(|\mathcal{S}_1|=n)  = \mathrm{P}(|\mathcal{S}_2|=K-n)$,  $F_{\mathcal{S}_1|K}(y) =    \frac{ F_{{\pi(j)}}(\min\{y,\xi\})}{  F_{{\pi(j)}}(\xi) }$,   $F_{\mathcal{S}_2|K}(y)
  =  \max\left\{0, \frac{F_{{\pi(j)}}(y) -F_{{\pi(j)}}(\xi) }{1-F_{{\pi(j)}}(\xi) }\right\}$,
$\tilde{\eta}_2 = \max\left\{
\frac{\frac{\tilde{\epsilon}_1}{\rho}}{\beta_1^2-\beta_2^2\tilde{\epsilon}_1}, \frac{\tilde{\epsilon}_2}{\rho \beta_2^2}\right\}$,   $\tilde{\eta}_1 =
\frac{\frac{\tilde{\epsilon}_1}{\rho}}{\beta_1^2-\beta_2^2\tilde{\epsilon}_1}$,  $\tilde{\epsilon}_k=2^{\tilde{R}_k}-1$, for $k\in\{1,2\}$, and $\tilde{R}_1$ and $\tilde{R}_2$ denote the targeted rates for the users selected to act as the weak and strong  users, respectively.
\end{lemma}
\begin{proof}
Since  there are $K\geq 2$ users in the sector, the outage probability experienced by the user chosen to act as the strong user in NOMA   can be expressed as follows:
\begin{align}\label{F5}
\mathrm{P}_{\mathcal{S}_2}^o =&
  F_{\mathcal{S}_2|K}(\tilde{\eta}_2)
  \mathrm{P}(|\mathcal{S}_2|\text{ is not empty})  \\\nonumber & +\mathrm{P}(|\mathcal{S}_2|\text{ is empty})  F_{\mathcal{S}_1|K}\left(\tilde{\eta}_2\right),
\end{align}  where $F_{\mathcal{S}_2|K}(\cdot)$ denotes the CDF of the   effective channel gain of a user randomly selected  from $\mathcal{S}_2$ and its expression will be evaluated later. The probability $ \mathrm{P}(|\mathcal{S}_2|\text{ is not empty})$ is equivalent to $\sum^{K}_{n=1}\mathrm{P}(|\mathcal{S}_1|=n)$.   Note that $F_{\mathcal{S}_1|K}\left(\tilde{\eta}_2\right)$, the CDF of the weak user's channel,  is used for the case of $|\mathcal{S}_2|=0$ since the base station will select one user randomly from $\mathcal{S}_1$ to act as the strong user with the targeted rate of $\tilde{R}_2$ for the NOMA transmission. Similarly, the outage probability experienced by the user selected to act as the weak user can be expressed as follows:
\begin{align}\label{F4}
\mathrm{P}_{\mathcal{S}_1}^o =&
  F_{\mathcal{S}_1|K}(\tilde{\eta}_1)
 \sum^{K}_{n=1}\mathrm{P}(|\mathcal{S}_1|\text{ is not empty})  \\\nonumber & +\mathrm{P}(|\mathcal{S}_1|\text{ is   empty})   F_{\mathcal{S}_2|K}\left(\tilde{\eta}_1\right),
\end{align}
where   the variables  are defined similarly to their counterparts in \eqref{F5}.

Given  that there are $K$ users in the sector, the probability for the case of $|\mathrm{S}_2|=n$ can be obtained as shown in the lemma,
which is due to the fact that all the users' channels are independent and identically distributed.

The CDF of the effective channel gain of a user randomly selected from $\mathcal{S}_1$ can be expressed as follows:
\begin{align}
F_{\mathcal{S}_1|K}(y) &=   \mathrm{P}\left(|\mathbf{h}_{j}^H   \mathbf{p}|^2<y~|~ X_{j}\in\mathcal{D}_{\theta}, j\in\mathcal{S}_1\right) \\
 &=   \nonumber  \frac{\underset{\mathcal{D}_{\theta}}{\int}\mathrm{P}\left(|\mathbf{h}_{j}^H   \mathbf{p}|^2<\min\{y,\xi\}~|~ X_{j}=x_{j}\right)dP_{X_{j}}}{\underset{\mathcal{D}_{\theta}}{\int}\mathrm{P}\left(|\mathbf{h}_{j}^H   \mathbf{p}|^2<\xi~|~ X_{j}=x_{j}\right)dP_{X_{j}}}.
\end{align}
Following steps similar to those in Section \ref{section 2}, the addressed CDF can be obtained as follows:
\begin{align}\label{F1}
F_{\mathcal{S}_1|K}(y) =&    \frac{ F_{{\pi(j)}}(\min\{y,\xi\})}{  F_{{\pi(j)}}(\xi) },
\end{align}
for $\xi> 0$.

On the other hand, the CDF of the effective channel gain of a  user randomly selected from $\mathcal{S}_2$ can be expressed as follows:
\begin{align}
F_{\mathcal{S}_2|K}(y) &=   \mathrm{P}\left(|\mathbf{h}_{j}^H   \mathbf{p}|^2<y~|~ X_{j}\in \mathcal{D}_{\theta}, j\in\mathcal{S}_2\right) \\\nonumber
  &=    \frac{\underset{\mathcal{D}_{\theta}}{\int}\mathrm{P}\left(\xi<|\mathbf{h}_{j}^H   \mathbf{p}|^2<y~|~ X_{j}=x_{j}\right)dP_{X_{j}}}{\underset{\mathcal{D}_{\theta}}{\int}\mathrm{P}\left(|\mathbf{h}_{j}^H   \mathbf{p}|^2>\xi~|~ X_{j}=x_{j}\right)dP_{X_{j}}},
\end{align}
if $y>\xi$, otherwise $F_{\mathcal{S}_2|K}(y)=0$.
Again following steps similar to those in Section \ref{section 2}, this CDF can be found as follows:
\begin{align}\label{F2}
F_{\mathcal{S}_2|K}(y)
  &=   \frac{F_{{\pi(j)}}(y) -F_{{\pi(j)}}(\xi) }{1-F_{{\pi(j)}}(\xi) },
\end{align}
 if $\xi<y$, otherwise $F_{\mathcal{S}_2|K}(y)=0$. Substituting \eqref{F1} and \eqref{F2} into \eqref{F4} and \eqref{F5}, the outage probabilities in the lemma can be obtained and the proof is complete.
\end{proof}

By using the outage probabilities obtained in the above lemma, one can easily find an expression for the outage sum rate, which is omitted here due to space limitations.

Obviously the choice of $\xi$ will have a significant  impact on  the performance of the addressed one-bit feedback scenario. To    investigate this impact,  we will first study the impact of $\xi$ on the CDFs, $F_{\mathcal{S}_1|K}(\tilde{\eta}_1)$ and $F_{\mathcal{S}_2|K}(\tilde{\eta}_2)$.
\subsubsection{The impact of the threshold  on $F_{\mathcal{S}_k|K}(\tilde{\eta}_k)$}
%Obviously $\eta_1\xi<\eta_2$ which means that at high SNR, $\xi$ also goes to zero. Otherwise ???

Because   $\tilde{\eta}_1$ approaches zero at high SNR, $\min\{\tilde{\eta}_1,\xi\}$ will also approach zero at high SNR, with a  rate of decaying no smaller than $\tilde{\eta}_1$.
Note that the outage probability of the user selected to act as the NOMA weak user  is related to $F_{\mathcal{S}_1|K}(\tilde{\eta}_1)$, which can be approximated at high SNR as follows:
\begin{align}
F_{\mathcal{S}_1|K}(\tilde{\eta}_1) =& \frac{ F_{{\pi(j)}}(\min\{\tilde{\eta}_1,\xi\})}{  F_{{\pi(j)}}(\xi) }\\\nonumber
\approx&
\frac{\min\{\tilde{\eta}_1,\xi\} }{2M   \gamma(2, R_{\mathcal{D}}\phi ) F_{{\pi(j)}}(\xi) }  \left(2+\frac{\pi^2 M^2 \Delta^2}{18}\right)\\\nonumber &\times
    \left( \gamma(2,R_{\mathcal{D}}\phi)+ \phi^{-\alpha}\gamma(\alpha+2,R_{\mathcal{D}}\phi)\right).
\end{align}
In this paper, we are interested in the      following two choices of $\xi$.
\begin{itemize}
\item If $\xi$ is a constant and not a function of the transmit SNR, $\rho$,  the following holds at high SNR:
\begin{align}\label{xxx1}
F_{\mathcal{S}_1|K}(\tilde{\eta}_1) \dot\sim \frac{1}{\rho}.
\end{align}

\item If $\xi$   decreases at a  rate of $\frac{1}{\rho^x}$, i.e., $\xi\dot\sim \frac{1}{\rho^x}$, $x>0$,  we have the following approximation:
\begin{align}\label{xxx2}
F_{\mathcal{S}_1|K}(\tilde{\eta}_1) &\approx  \frac{  \min\{\tilde{\eta}_1,\xi\}}{   \xi }.
\end{align}

\end{itemize}
%
%
%As can be seen from the above expression, in order to avoid the case of $F_{\mathcal{S}_1|K}(\tilde{\eta}_1) =1$, it is important to have
%\begin{align}\label{constraint 1}
%\xi<\tilde{\eta}_1.
%\end{align}
On the other hand, the impact of $\xi$ on  $F_{\mathcal{S}_2|K}(\tilde{\eta}_2)$ can be demonstrated as follows.
\begin{itemize}
\item If $\xi$ is a constant, $F_{\mathcal{S}_2|K}(\tilde{\eta}_2)=0$ at high SNR,  since $\tilde{\eta}_2$ approaches zero at high SNR and hence
\[
\mathrm{P}\left(\xi<|\mathbf{h}_{j}^H   \mathbf{p}|^2<\tilde{\eta}_2~|~ X_{j}=x_{j}\right)=0.
\]

\item If $\xi\dot\sim \frac{1}{\rho^x}$, $x>0$, we  have the following approximation:
\begin{align}
F_{\mathcal{S}_2|K}(\tilde{\eta}_2) =&   \frac{F_{{\pi(j)}}(\tilde{\eta}_2) -F_{{\pi(j)}}(\xi) }{1-F_{{\pi(j)}}(\xi) }\\\nonumber \approx& \frac{(\tilde{\eta}_2-\xi) }{2M   \gamma(2, R_{\mathcal{D}}\phi )   }  \left(2+\frac{\pi^2 M^2 \Delta^2}{18}\right)\\\nonumber &\times
    \left( \gamma(2,R_{\mathcal{D}}\phi)+ \phi^{-\alpha}\gamma(\alpha+2,R_{\mathcal{D}}\phi)\right),
\end{align}
if $\xi<\tilde{\eta}_2$, otherwise $F_{\mathcal{S}_2|K}(\tilde{\eta}_2)=0$.
\end{itemize}

\subsubsection{The impact of $\xi$ on the users' outage probabilities, $\mathrm{P}_{\mathcal{S}_k}^o$}
We first focus on the user selected to act as the weak user, whose diversity gain is shown in the following lemma.
\begin{lemma}
For the two considered choices of the threshold, i.e., either $\xi\dot\sim\frac{1}{\rho^x}$ or $\xi$ is a constant, the diversity order of the user selected to act as the weak user is always one.
\end{lemma}
\begin{proof}
For notational simplicity, let $c_2=\frac{ \left(2+\frac{\pi^2 M^2 \Delta^2}{18}\right)
     }{2M   \gamma(2, R_{\mathcal{D}}\phi )   }\left( \gamma(2,R_{\mathcal{D}}\phi)+ \phi^{-\alpha}\gamma(\alpha+2,R_{\mathcal{D}}\phi)\right)$. It can be shown   that the probability of having $n$ users in group $\mathcal{S}_2$ can be approximated as follows:
\begin{align}
\mathrm{P}(|\mathcal{S}_2|=n)  =& {K \choose n} \left(F_{{\pi(j)}}(\xi)\right)^{K-n} \left(1-F_{{\pi(j)}}(\xi)\right)^n  \\\nonumber \approx&  {K \choose n}  c_2^{K-n}  \xi ^{K-n},
\end{align}
if $\xi$ approaches zero. If $\xi$ is not a function of $\rho$, neither is this probability.

\begin{itemize}
\item If $\xi$ is a constant,  the outage probability of  the user selected to act as the weak user can be simplified as follows:
 \begin{align}
\mathrm{P}_{\mathcal{S}_1}^o =&
  F_{\mathcal{S}_1|K}(\tilde{\eta}_1)
 \sum^{K}_{n=1}\mathrm{P}(|\mathcal{S}_1|=n)   \dot\sim \frac{1}{\rho},
 \end{align}
since $F_{\mathcal{S}_2|K}\left(\tilde{\eta}_1\right)=0$, $\sum^{K}_{n=1}\mathrm{P}(|\mathcal{S}_1|=n)$ is a constant and $  F_{\mathcal{S}_1|K}(\tilde{\eta}_1)\dot\sim \frac{1}{\rho}$ as explained in \eqref{xxx1}. Therefore the user's diversity order is one for this choice of $\xi$.

\item If $\xi\dot\sim \frac{1}{\rho^x}$, $x> 0$, the outage probability for the user selected   to act as the weak user can be approximated as follows:
 \begin{align}
\mathrm{P}_{\mathcal{S}_1}^o =&
  F_{\mathcal{S}_1|K}(\tilde{\eta}_1)
 \sum^{K}_{n=1}\mathrm{P}(|\mathcal{S}_1|=n)  \\\nonumber &+\mathrm{P}(|\mathcal{S}_1|=0)  F_{\mathcal{S}_2|K}\left(\tilde{\eta}_1\right)\\\nonumber \approx&
   \frac{  \min\{\tilde{\eta}_1,\xi\}}{   \xi } \sum^{K}_{n=1}
 {K \choose n}  c_2^{n}  \xi ^{n}  +  c_2\max\{0,\tilde{\eta}_1-\xi\}\\\nonumber
 \approx&
   \frac{  \min\{\tilde{\eta}_1,\xi\}}{   \xi }
 K  c_2  \xi   +  c_2\max\{0,\tilde{\eta}_1-\xi\},
\end{align}
which is always at the order of $\frac{1}{\rho}$ as explained in the following.

If $x>1$, $\min\{\tilde{\eta}_1,\xi\}=\xi$ and $\max\{0,\tilde{\eta}_1-\xi\}\approx \tilde{\eta}_1$. These two observations lead to the following approximation:
 \begin{align}
\mathrm{P}_{\mathcal{S}_1}^o
 \approx&
    \xi
 K  c_2      +  c_2\tilde{\eta}_1  \dot\sim \frac{1}{\rho},
\end{align}
since $\tilde{\eta}_1$ is dominant.

If $x=1$, we have the following approximation:
 \begin{align}
\mathrm{P}_{\mathcal{S}_1}^o
 \approx&
     \min\{\tilde{\eta}_1,\xi\}
 K  c_2     +  c_2\max\{0,\tilde{\eta}_1-\xi\}\dot\sim \frac{1}{\rho},
\end{align}
since both $\min\{\tilde{\eta}_1,\xi\}$ and $|\tilde{\eta}_1-\xi|$   are at the order of $\frac{1}{\rho}$.

Further, if $0< x<1$, we have the following approximation:
 \begin{align}
\mathrm{P}_{\mathcal{S}_1}^o
 \approx&
       \tilde{\eta}_1
 K  c_2      \dot\sim \frac{1}{\rho}.
\end{align}
Therefore, we can conclude that,   as long as $\xi\dot\sim \frac{1}{\rho^x}$, $x>0$, the diversity order of the user selected to act as the weak user is one.

\end{itemize}
Since the user's diversity order is one for both cases,   the proof is complete.
\end{proof}

However, the impact of $\xi$ on $\mathrm{P}_{\mathcal{S}_2}^o$ is more complicated as illustrated in the following:
\begin{itemize}
\item
If   $\xi$ is a constant, the diversity order of the user selected to act as the strong user is one, since
\begin{align}\nonumber
\mathrm{P}_{\mathcal{S}_2}^o =& \mathrm{P}(|\mathcal{S}_2|=0) F_{\mathcal{S}_1|K}\left(\tilde{\eta}_2\right)\dot\sim \frac{1}{\rho},
\end{align}
which is due to  the following facts:   $F_{\mathcal{S}_2|K}(\tilde{\eta}_2)=0$,  $\mathrm{P}(|\mathcal{S}_2|=0)$ is a constant and $F_{\mathcal{S}_1|K}\left(\tilde{\eta}_2\right)\dot\sim \frac{1}{\rho}$ as explained in \eqref{xxx1}.

\item
If   $\xi\dot\sim \frac{1}{\rho^x}$, $x\geq 1$, the outage probability for the user selected to act as the strong user can be approximated as follows:
\begin{align}\nonumber
\mathrm{P}_{\mathcal{S}_2}^o =&
  F_{\mathcal{S}_2|K}(\tilde{\eta}_2)
 \sum^{K}_{n=1}\mathrm{P}(|\mathcal{S}_2|=n) \\\nonumber & +\mathrm{P}(|\mathcal{S}_2|=0)    F_{\mathcal{S}_1|K}\left(\tilde{\eta}_2\right)\\ \label{diversity gain} \approx&
  c_2\max\{0,\tilde{\eta}_2-\xi\}  +c_2^{K}  \xi ^{K}   \frac{  \min\{\tilde{\eta}_2,\xi\}}{   \xi }.
\end{align}
As can be seen from \eqref{diversity gain}, the choice of the threshold $\xi$ has significant  impact on the achievable diversity gain. For example, a full diversity gain of $K$ can be obtained by using the following choice of $\xi$:
\begin{align}
\xi = \tilde{\eta}_2 - \frac{1}{\rho^K}.
\end{align}
\end{itemize}
\section{Random Beamforming: A Multiple-Beam Case}
\subsection{System Model and Outage Performance}Consider a scenario in which   the base station will form $N$, $1<N\leq M$, orthonormal   beams, denoted by $\mathbf{p}_m$, $1\leq m\leq N$, where $\mathbf{p}_m^H\mathbf{p}_m=1$ and $\mathbf{p}_m^H\mathbf{p}_n=0$ if $m\neq n$.  These beamforming vectors are predefined, and it is assumed that they are known to the base station and the users prior to transmission. Following \cite{7279196} and \cite{7397837}, these $N$ orthonormal beamforming vectors can be constructed as follows:
\begin{align}
\mathbf{p}_m = \mathbf{a}\left(\zeta +\frac{2(m-1)}{N}\right),
\end{align}
for $1\leq m \leq N$, where $\zeta$ denotes a random variable following a uniform distribution between $-1$ and $1$. For notational  simplicity, we denote $\zeta +\frac{2(m-1)}{N}$ by $\bar{\theta}_m$. Again  this beamformer can also be  viewed   as   a special case of the  hybrid precoding design in \cite{7445130}, in which  the fully-connected architecture is used with $N$ radio frequency chains, $M$ antennas and a digital precoding matrix set as an identity matrix.

 Prior to   downlink transmission, the base station will first broadcast pilot signals on these $N$ orthogonal beams. Similarly to $\mathcal{D}_{\theta}$, define  $\mathcal{D}_{\theta_m}$ as the wedge-shaped  sector around $\bar{\theta}_m$ with a central angle of $2\Delta$, as shown in Fig. \ref{system_model}.    Only the users that fall into the sector $\mathcal{D}_{\theta_m}$ will participate in the NOMA transmission on beam $m$. Denote the number of users in $\mathcal{D}_{\theta_m}$ by $K_m$ and the $k^{\rm th}$ user's channel by $\mathbf{h}_{m,k}$, $1\leq k \leq K_m$.
   Each user will measure its effective channel gain on its corresponding beam, where  user $k$'s effective channel gain on the $m$-th beam  is given by $
 |\mathbf{h}_{m,k}^H\mathbf{p}_m|^2$.
Without loss of generality,   we assume that the base station schedules   user $i$ and user $j$ on beam $m$, to act as the weak and strong users, respectively.

  Therefore, the base station will superimpose two users' messages   on each of the  $N$ beams  as follows:
 \begin{align}
\sum^{N}_{m=1} \mathbf{p}_m \left(\beta_{m,1} s_{m,i}+\beta_{m,2} s_{m,j}\right),
 \end{align}
 where $\beta_{m,1}^2+\beta_{m,2}^2=1$.

  Therefore,  user $j$ on beam $m$ will receive the following observation:
 \begin{align}
 y_{m,j} =& \mathbf{h}_{m,j}^H \sum^{N}_{n=1} \mathbf{p}_n \left(\beta_{n,1} s_{n,i}+\beta_{n,2} s_{n,j}\right)+n_{m,j}
\\\nonumber =&\mathbf{h}_{m,j}^H   \mathbf{p}_m \left(\beta_{m,1} s_{m,i}+\beta_{m,2} s_{m,j}\right) \\\nonumber &+ \mathbf{h}_{m,j}^H \sum^{N}_{n=1,n\neq m}   \mathbf{p}_n \left(\beta_{n,1} s_{n,i}+\beta_{n,2} s_{n,j}\right) +n_{m,j},
 \end{align}
 where $n_{m,j}$ denotes additive Gaussian noise.

 User $j$ on beam $m$  will first decode the message to user $i$ in the same pair,  and then remove this message from its observation. Such SIC needs to be carried out before its own message is decoded.   As a result, the SINR for  user $j$ on beam $m$  to decode its partner's message   can be expressed as follows:
 \begin{align}
  {\text{SINR}}_{m, i\rightarrow j} =&\frac{|\mathbf{h}_{m,j}^H   \mathbf{p}_m|^2 \beta^2_{m,1} }{|\mathbf{h}_{m,j}^H   \mathbf{p}_m|^2 \beta^2_{m,2} +\underset{ n\neq m}{\sum} |\mathbf{h}_{m,j}^H \mathbf{p}_n|^2 +\frac{1}{\rho}} .
 \end{align}
Define   $R_{m,1}$ as the targeted rate for user $i$  on beam $m$ and $\epsilon_{m,1}=2^{R_{m,1}} -1$,  where  $\epsilon_{m,2}$ and $R_{m,2}$ are  defined for user $j$ similarly. If ${\text{SINR}}_{m, i\rightarrow j}\geq \epsilon_{m,1}$, intra-group interference can be cancelled and the user can decode its own information with the following SINR:
\begin{align}
 \text{SINR}_{m,j} =&\frac{|\mathbf{h}_{m,j}^H   \mathbf{p}_m|^2 \beta^2_{m,2} }{\underset{ n\neq m}{\sum} |\mathbf{h}_{m,j}^H \mathbf{p}_n|^2 +\frac{1}{\rho}} .
 \end{align}

 User $i$ on beam $m$ will decode its own message directly  with the following SINR:
 \begin{align}
  {\text{SINR}}_{m,i} =&\frac{|\mathbf{h}_{m,i}^H   \mathbf{p}_m|^2 \beta^2_{m,1} }{|\mathbf{h}_{m,i}^H   \mathbf{p}_m|^2 \beta^2_{m,2} +\underset{ n\neq m}{\sum} |\mathbf{h}_{m,i}^H \mathbf{p}_n|^2 +\frac{1}{\rho}} .
 \end{align}

Different from the case with one beam, the users' SINRs are functions  not only of $|\mathbf{h}_{m,i}^H   \mathbf{p}_m|^2$ but also of $|\mathbf{h}_{m,i}^H   \mathbf{p}_n|^2$, $n\neq m$. In conventional non-NOMA scenarios, users can be scheduled according to their SINRs, i.e., the user with the strongest SINR on beam $m$ will be selected to be served on this beam. However, in the addressed scenario, one user can have two different SINR functions. For example,  user $j$'s performance depends  on two different SINR functions,  ${\text{SINR}}_{m, i\rightarrow j}$ and $\text{SINR}_{m,j}$. For the  purpose of illustration, we focus on  a simple user scheduling scheme based on distances, a strategy similar to the one proposed  in Section \ref{subsection ditance}.  Therefore we can order these users who will participate in the NOMA transmission on beam $m$ as follows:
 \begin{align}
d_{m,1}\leq \cdots \leq d_{m,K_m}.
 \end{align}
  Furthermore suppose that user $i$ has a distance larger than that of user $j$, i.e., $i>j$.

The outage probability experienced by user $j$ can be expressed as follows:
\begin{align}
\mathrm{P}_{m,j}^o = 1-\mathrm{P}\left( {\text{SINR}}_{m, i\rightarrow j} > \epsilon_{m,1}, \text{SINR}_{m,j}>\epsilon_{m,2} \right).
\end{align}
Again applying the mmWave channel model, ${\text{SINR}}_{m, i\rightarrow j} $, can be written as follows:
 \begin{align}
  {\text{SINR}}_{m, i\rightarrow j} =& \frac{|a_{m,j}|^2   }{{(1+ d_{m,j}^{\alpha})  }}F_M\left(\pi[\bar{\theta}_m - \theta_{m,j}]\right) \beta^2_{m,1} \\\nonumber &\times \left(\frac{|a_{m,j}|^2   }{{(1+ d_{m,j}^{\alpha})  }}F_M\left(\pi[\bar{\theta}_m - \theta_{m,j}]\right) \beta^2_{m,2}\right. \\\nonumber &\left.+\underset{ n\neq m}{\sum} \frac{|a_{m,j}|^2   }{{(1+ d_{m,j}^{\alpha})  }}F_M\left(\pi[\bar{\theta}_n - \theta_{m,j}]\right) +\frac{1}{\rho}\right)^{-1} .
 \end{align}
 Similarly,  $\text{SINR}_{m,j}$, can be expressed as follows:
 \begin{align}
 \text{SINR}_{m,j} =&\frac{\frac{|a_{m,j}|^2   }{{(1+ d_{m,j}^{\alpha})  }}F_M\left(\pi[\bar{\theta}_m - \theta_{m,j}]\right)  \beta^2_{m,2} }{\underset{ n\neq m}{\sum} \frac{|a_{m,j}|^2   }{{(1+ d_{m,j}^{\alpha})  }}F_M\left(\pi[\bar{\theta}_n - \theta_{m,j}]\right)  +\frac{1}{\rho}} .
 \end{align}

Unlike those SINR functions in the previous sections, the SINRs for the case with multiple beams become more complicated. An interesting observation is that the three factors in the numerator and denominator of ${\text{SINR}}_{m, i\rightarrow j}$ share the same fading coefficient. In this case, the outage probability of user $j$ on beam $m$ can be expressed as shown in \eqref{figurex} at the tope of the following page,
\begin{figure*}
\begin{align}\label{figurex}
\mathrm{P}_{m,j}^o  =&1-
 \int^{\bar{\theta}+\Delta}_{\bar{\theta}-\Delta}\int^{R_{\mathcal{D}}}_0 \exp \{
 -\max \{
  \frac{\frac{\epsilon_{m,1}}{\rho}(1+ d_{m,j}^{\alpha})  }{
  F^m_{j,m} \beta^2_{m,1}
  - \epsilon_{m,1}F^m_{j,m} \beta^2_{m,2}
  -\underset{ n\neq m}{\sum}  \epsilon_{m,1}F^m_{j,n}
  }  ,\\\nonumber &
  \frac{\frac{\epsilon_{m,2}}{\rho}{(1+ d_{m,j}^{\alpha})  }}{
   F^m_{j,m}  \beta^2_{m,2}
  -\underset{ n\neq m}{\sum}  \epsilon_{m,2}F^m_{j,n}
  }
    \}
   \} \frac{f_{d_j}(r)}{2\Delta} drd\theta,
\end{align}
\end{figure*}
if $
  F^m_{j,m} \beta^2_{m,1}>
   \epsilon_{m,1}F^m_{j,m} \beta^2_{m,2}
  +\underset{ n\neq m}{\sum}  \epsilon_{m,1}F^m_{j,n} $ and $F^m_{j,m}  \beta^2_{m,2}
  >\underset{ n\neq m}{\sum}  \epsilon_{m,2}F^m_{j,n}  $, otherwise the outage probability will be always one,
where $F^m_{j,n}\triangleq F_M\left(\pi[\bar{\theta}_n - \theta_{m,j}]\right) $. The outage probability at user $i$ can be obtained similarly.

\subsection{Asymptotic Performance Analysis}
 Without loss of generality, we focus on the first beam, i.e., $m=1$. In this case, the factor $F^m_{j,n}$ can be written as follows:
\begin{align}
F^1_{j,n}&=  F_M\left(\pi[ \theta_{1,j}-\bar{\theta}_n]\right) \\\nonumber & =  F_M\left(\pi\left[ \theta_{1,j}-\bar{\theta}_1 -\frac{2(n-1)}{N}\right]\right),
\end{align}
where $2\leq n \leq N$.
We   have the following Taylor series approximation:
\begin{align}
F^1_{j,n} &= \sum^{\infty}_{l=0}  F_M^{(l)}\left( - \frac{2(n-1)\pi}{N} \right)\frac{(\theta_{1,j}-\bar{\theta}_1 )^l}{l!},
\end{align}
where $F^{(l)}_M(x)$ denotes the $n$-th derivative of $F_M(x)$. Here we assume that the derivatives, $F_M^{(l)}\left( - \frac{2(n-1)\pi}{N} \right)$, exist for all orders.  Assume that the beams are separated with sufficient gaps, and one can expect that $ F_M\left( - \frac{2(n-1)\pi}{N} \right)\rightarrow 0$, for $2\leq n \leq N$.
Further assuming $\Delta \rightarrow 0$,  $(\bar{\theta}_n - \theta_{1,j})$ approaches zero, which means
\begin{align}
F^1_{j,n} \approx&     F_M\left( - \frac{2(n-1)\pi}{N} \right) \\\nonumber & +  F_M^{(1)}\left( - \frac{2(n-1)\pi}{N} \right) (\bar{\theta}_1 - \theta_{1,j}),
\end{align}
where $F_M^{(1)}(x)=\frac{\sin Mx}{1-\cos(x)}-\frac{(1-\cos Mx)\sin x}{M(1-\cos x)^2}$.

 Therefore the sum of the interference terms in the SINR expressions can be approximated as follows:
 \begin{align}
 \underset{ n\neq 1}{\sum}   F^1_{j,n} \approx c_2+c_3 (\bar{\theta}_1 - \theta_{1,j}),
 \end{align}
 where $c_2=  \underset{ n\neq 1}{\sum}   F_M\left(-  \frac{2(n-1)\pi}{N} \right)$ and $c_3= \underset{ n\neq 1}{\sum}    F_M^{(1)}\left(  -\frac{2(n-1)\pi}{N} \right) $.
For the case  $n=1$, we have
 \begin{align}
  F^1_{j,1} \approx M, %-(\bar{\theta}_1 - \theta_{1,j})
 \end{align}
which is obtained from \eqref{approximtion x}.
As a result, at high SNR, the outage probability experienced by user $i$ can be expressed as follows:
\begin{align}\nonumber
\mathrm{P}_{1,i}^o  \approx&\frac{1}{2\Delta}
 \int^{\bar{\theta}+\Delta}_{\bar{\theta}-\Delta}\int^{R_{\mathcal{D}}}_0
    \frac{\frac{\epsilon_{1,1}}{\rho}(1+ d_{1,i}^{\alpha})  }{
  F^1_{i,1} \beta^2_{1,1}
  - \epsilon_{1,1}F^1_{i,1} \beta^2_{1,2}
  -\underset{ n\neq 1}{\sum}  \epsilon_{1,1}F^1_{i,n}
  }\\\nonumber &\times  f_{d_i}(r) drd\theta
   \\\nonumber
   \approx& MQ_{1,i}
 \int^{ \Delta}_{ -\Delta}
  \frac{\frac{\epsilon_{1,1}}{\rho}   }{  ( Mc_4-c_2 \epsilon_{1,1}) -  c_3 \epsilon_{1,1}y
  } dy\\
            \approx&
  \frac{2\Delta MQ_{1,i}\epsilon_{1,1}}{\rho( Mc_4-c_2 \epsilon_{1,1})},
\end{align}
where $c_4= \beta^2_{1,1}-\epsilon_{1,1}  \beta^2_{1,2}$. The outage probability for user $j$ can be obtained similarly. As a result, following steps similar to those in Section \ref{subsection ditance}, the outage sum rate and the outage probabilities can be obtained.

\begin{figure}
  \centering
  \subfigure[Sum Rates]{
  \includegraphics[width=0.4\textwidth]{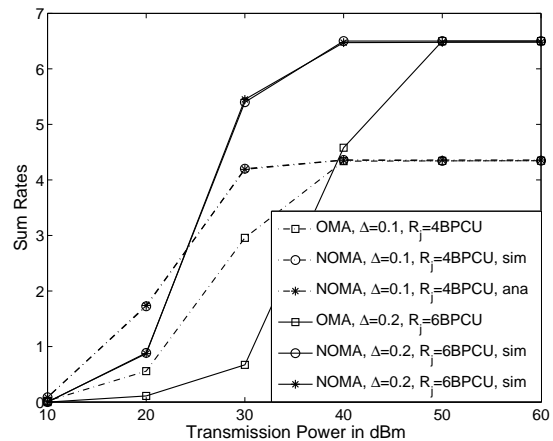}}
  \subfigure[Outage Probabilities $R_2=6$ BPCU and $K=5$]{
  \includegraphics[width=0.4\textwidth]{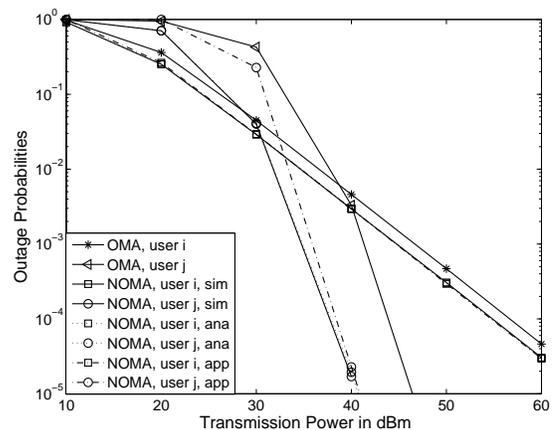}}\vspace{-0.5em}
    \caption{The performance of mmWave-NOMA and mmWave-OMA with  perfect CSI.  $M=4$,  $\lambda=1$, $\Delta=0.1$,      $R_i=0.5$ BPCU, $i=1$ and $j=K$. }
     \label{fig1}\vspace{-1.5em}
     \end{figure}

\section{Numerical Studies}
In this section, the performance of the proposed mmWave-NOMA transmission schemes are evaluated by using computer simulations, where the accuracy of the developed analytical results will also be verified.  The path loss exponent is set as $\alpha=2$, since line-of-sight links are focused. The radius  of $\mathcal{D}$  is $R_{\mathcal{D}}=10$m,  the noise power is $-30$dBm, the blockage parameter is set as $\phi=0.1$, and $\beta_i^2=\frac{3}{4}$ and $\beta_j^2=\frac{1}{4}$ are used as the NOMA power allocation coefficients. It is worth pointing out that our analytical results are developed for arbitrary choices of these parameters, and using other choices of these parameters will lead to conclusions similar to those drawn  in this section.

In Fig. \ref{fig1}, the performance of the proposed random beamforming scheme in mmWave-NOMA systems with perfect CSI is studied, where the mmWave-OMA scheme is used as a benchmark. Fig.\ref{fig1}.(a) shows the outage sum rates achieved by the two MA schemes, and Fig. \ref{fig1}.(b) shows  the outage probabilities of the two transmission schemes. As can be observed from Fig. \ref{fig1}.(a), the use of NOMA can yield a significant sum rate gain over the OMA scheme, and this gain increases  when  the targeted data rate of the strong user is increased. For example, for $R_j=4$ bits per channel use (BPCU), the gain of mmWave-NOMA over mmWave-OMA is $1$ BPCU, when the transmission power of the base station is $30$ dBm. When $R_j$ is increased to $6$ BPCU, the performance gain of the NOMA scheme over OMA becomes $5$ BPCU.  On the other hand, Fig. \ref{fig1}.(b) shows that the mmWave-NOMA scheme can also effectively reduce the outage probability, compared to OMA, particularly for the user with the stronger  channel. It is also important to point out that the developed approximation results for the sum rate and the outage probabilities are tight at high SNR, and the developed exact expressions match  the simulation results perfectly.

\begin{figure}
  \centering
  \subfigure[Sum Rates]{
  \includegraphics[width=0.4\textwidth]{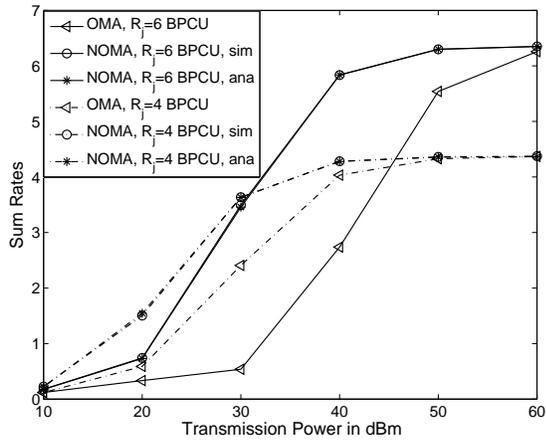}}
  \subfigure[Outage Probabilities $R_2=6$ BPCU]{
  \includegraphics[width=0.4\textwidth]{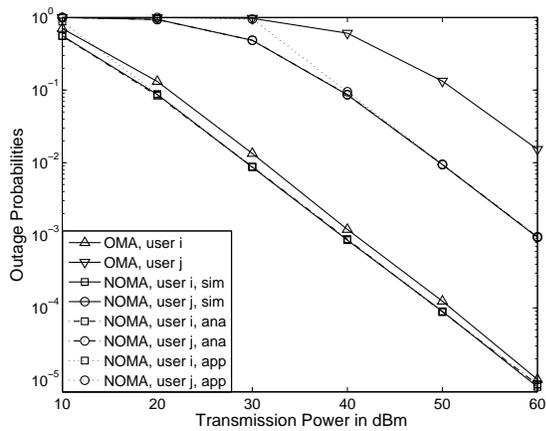}}\vspace{-0.5em}
    \caption{The performance of mmWave-NOMA and mmWave-OMA by using the  distance information only.  $M=4$, $\lambda=1$,  $\Delta=0.1$,   $R_i=0.5$ BPCU, $i=4$ and $j=1$. }\vspace{-1.5em}
     \label{fig2}
     \end{figure}

In Fig. \ref{fig2}, the performance of the mmWave-NOMA and mmWave-OMA schemes is compared, for the situation in which  the base station has access to the users' distance information only. The trivial cases in which  the $i$-th and $j$-th nearest nodes do not exist can cause error floors to the outage probabilities. Therefore,  we slightly change the definition of the outage probability by   counting only the cases in which  the two nodes can be found in $\mathcal{D}_{\theta}$. Take the outage probability for user $i$ as an example. The outage probability curves are obtained by using $\frac{n_3}{n_1-n_2}$, where  $n_1$ denotes the total number of simulations, $n_2$ denotes the number of events in which user $i$ cannot be found in $\mathcal{D}_{\theta}$, and $n_3$ denotes the number of outage events by excluding the outage events caused by the case in which  user $i$ cannot be found (i.e., $n_2$).  This is consistent with  \eqref{Fxx} since the   probability shown in the figure is equivalent to the following one
\begin{align}
\frac{F^o_{k} - \sum^{k-1}_{n=0}\mathrm{P}(K=n)}{1-\sum^{k-1}_{n=0}\mathrm{P}(K=n)}.
\end{align}
As can be observed from both figures,   the use of NOMA can yield  a significant performance gain in the sum rate and effectively reduce the outage probability, compared to the OMA scheme, even if only the distance information is available to the base station. Again both figures also demonstrate the accuracy of the developed analytical results.

\begin{figure}
  \centering
  \subfigure[Sum Rates]{
  \includegraphics[width=0.4\textwidth]{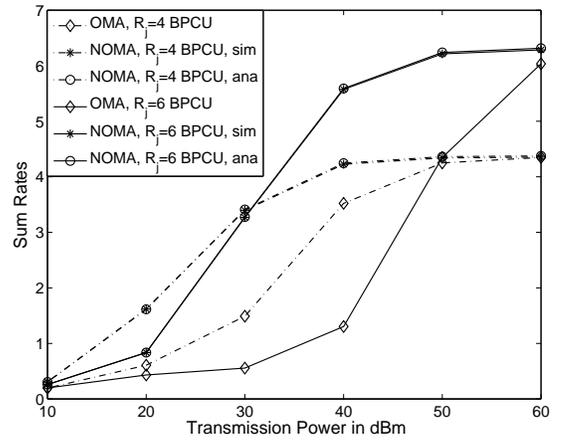}}
  \subfigure[Outage Probabilities $R_2=6$ BPCU and $K=5$]{
  \includegraphics[width=0.4\textwidth]{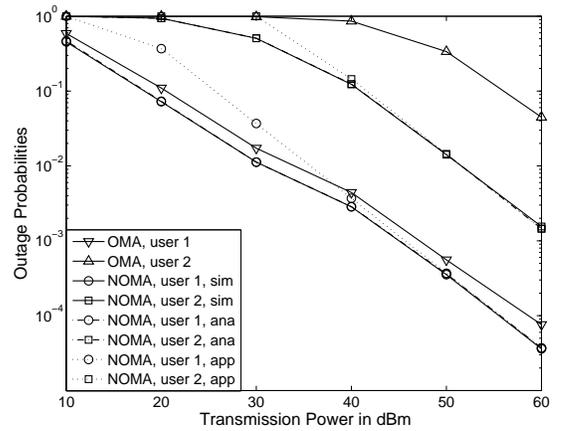}}\vspace{-0.5em}
    \caption{The performance of mmWave-NOMA and mmWave-OMA with one-bit feedback. $M=4$,  $\lambda=1$,   $\Delta=0.1$,   and $R_i=0.5$ BPCU. The threshold is set as $\frac{1}{2}(\tilde{\eta}_1 +\tilde{\eta}_2)$.}
     \label{fig3}\vspace{-1.5em}
     \end{figure}

     \begin{figure}[!t]
\centering
\includegraphics[width=0.4\textwidth]{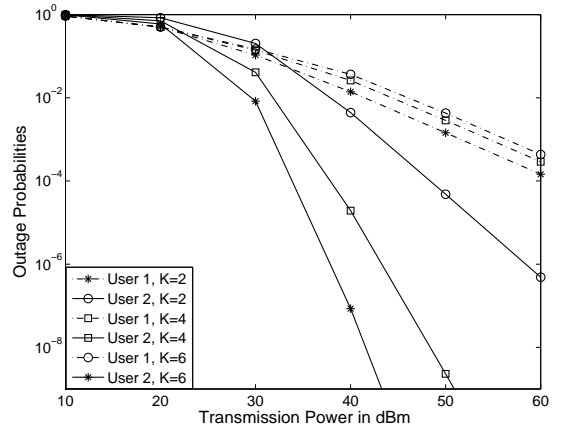}\vspace{-0.5em}
\caption{The impact of the threshold on the performance of the one-bit feedback scheme.   $M=4$,  $\lambda=1$, $\Delta=0.1$,    $R_1=1.5$ BPCU and $R_2=4$ BPCU. The threshold is set as $(\tilde{\eta}_2-\frac{1}{\rho^K})$.} \label{fig4}\vspace{-1.5em}
\end{figure}

 Fig. \ref{fig3} shows the sum rate and the outage probabilities achieved by the one-bit feedback scheme, and Fig. \ref{fig4} shows the impact of the threshold $\xi$ on   the users' outage probabilities and   diversity gains. Consistent with the previous figures, Fig. \ref{fig3} demonstrates that the use of NOMA can significantly improve the performance of mmWave communications with random beamforming. Recall that the choice of the threshold has a significant  impact on the performance of the one-bit feedback scheme. As discussed in Section \ref{subsection 2}, the diversity gain of the strong user is particularly  sensitive to the choice of the threshold, and a choice of $\xi=\tilde{\eta}_j-\frac{1}{\rho^K}$ yields   a diversity gain of $K$, whereas the diversity gain of the weak user is always one for the discussed choices of $\xi$. Fig. \ref{fig4} clearly confirms these analytical results and  demonstrates the impact of $\xi$ on the diversity gain. For example,  the slope of  the strong user's outage probability curve becomes larger when increasing $K$, which demonstrates that the diversity gain of this user is an increasing function of $K$. On the other hand, the slope for the other user's outage probability curve is always the same, which shows that the diversity gain of the weak user is not sensitive to the choice of the threshold.

\begin{figure}
  \centering
  \subfigure[Sum Rates]{
  \includegraphics[width=0.4\textwidth]{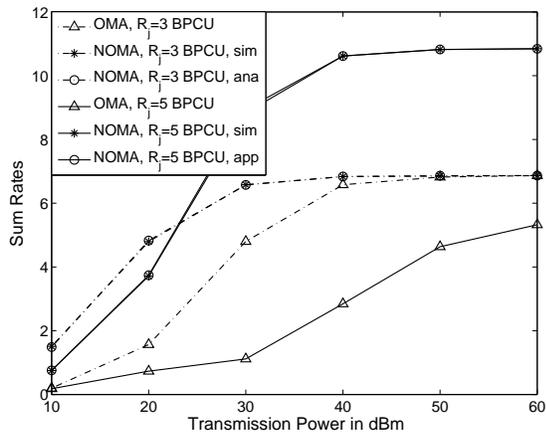}}
  \subfigure[Outage Probabilities $R_j=5$ BPCU]{
  \includegraphics[width=0.4\textwidth]{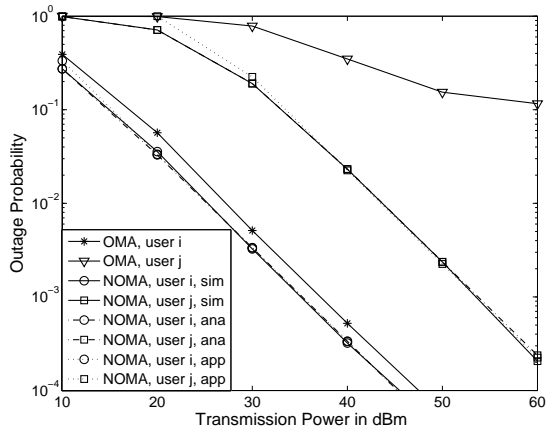}}\vspace{-0.5em}
    \caption{The performance of mmWave-NOMA and mmWave-OMA with multiple beams. $M=8$, $N=4$, $\lambda=10$,  $\Delta=0.01$,  $i=3$, $j=1$ and  $R_i=0.5$ BPCU. }
     \label{fig5}\vspace{-2.5em}
     \end{figure}

In Fig. \ref{fig5}, the performance of the proposed mmWave-NOMA scheme with multiple randomly generated beams is illustrated, where   OMA   is used as the benchmark again. Different from the previous cases with a single beam, the use of multiple beams means that users in the mmWave-NOMA system suffer more interference. Particularly, even if the strong user in a NOMA pair can use SIC to remove its parter's message, it still experiences interference from the users on other beams. However, the   fact   that mmWave propagation is highly directional  can   be used to effectively reduce such inter-beam interference.  The reason is that the inter-beam interference, $\underset{ n\neq m}{\sum} \frac{|a_{m,j}|^2   }{{(1+ d_{m,j}^{\alpha})  }}F_M\left(\pi[\bar{\theta}_n - \theta_{m,j}]\right) $,  is a function of the angle difference between a user's channel vector and the interference beams. With a choice of $\Delta=0.01$, i.e., the central angle is about $4$ degrees, the inter-beam interference is significantly suppressed, as shown in the two figures.   The superior performance of NOMA can also be  clearly demonstrated by the fact that the outage probability for the strong user in OMA cannot be reduced to zero, regardless of how large the transmission power is. On the other hand, the use of NOMA can reduce the outage probability rapidly by increasing the transmission power, which is due to the fact that NOMA can realize  better spectral efficiency.

Finally, we compare the mmWave-NOMA scheme with perfect CSI to the two schemes with limited CSI. Intuitively, the cases with limited CSI will result in some performance degradation, but the simulation results in Fig. \ref{fig6} indicate that the schemes  with limited feedback can yield an increase of the system throughput, as explained in the following.  Take a four-user case as an example, where the users are ordered as in \eqref{order}. Suppose that the perfect-CSI based scheme is to schedule user $1$ and user $2$, i.e., two users with poor channel conditions. Because of the ordering ambiguity caused by the use of partial CSI, the one-bit feedback scheme might schedule user $3$ and user $4$. According to the broadcast capacity region in \cite{Cover1991}, scheduling users with better channel conditions yields a larger sum rate, which means that it is possible for the schemes with partial CSI to outperform the one with perfect CSI.  Fig.\ref{fig6} clearly demonstrates this phenomenon.   For example, given $K$ users, when the   user with the worst channel condition is paired with the  user with the second worst channel condition. The schemes with limited  feedback can outperform the scheme with   perfect CSI, when the transmission power is $20$ dBm.  It is worth pointing out that a similar    observation has been previously reported in \cite{Zhiguo_massive} in the context of massive MIMO.

     \begin{figure}[!t]
\centering
\includegraphics[width=0.4\textwidth]{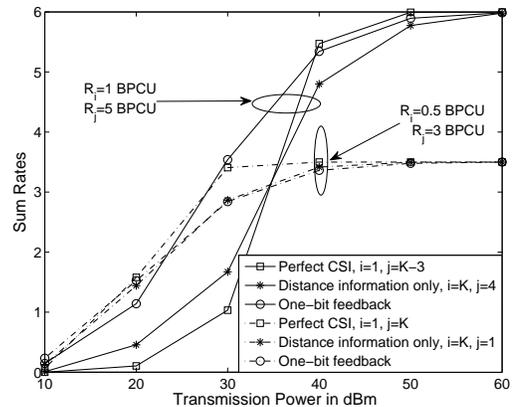}\vspace{-0.5em}
\caption{Sum-rate comparison between the proposed random beamforming transmission schemes.    $M=4$, $K=5$, $\lambda=1$, and $\Delta=0.4$. The threshold is set as $\frac{1}{2}(\tilde{\eta}_1+\tilde{\eta}_2)$.} \label{fig6} \vspace{-1.5em}
\end{figure}

%     \begin{figure}[!t]
%\centering
%\includegraphics[width=0.4\textwidth]{comx3.eps}\vspace{-1em}
%\caption{Outage comparison between the proposed random beamforming transmission schemes.    $M=4$, $K=5$, $\lambda=1$, $R_1=1$ BPCU and $R_2=5$ BPCU and $\Delta=0.4$. The threshold is set as $\frac{1}{2}(\tilde{\eta}_1+\tilde{\eta}_2)$.} \label{fig7} \vspace{-2.5em}
%\end{figure}
\section{Conclusions}
In this paper, we have  investigated the coexistence between NOMA and mmWave communications. We have first considered   the application of random beamforming to the addressed  mmWave-NOMA scenario, by focusing on the case with a single beam generated at the base station.  Stochastic geometry has been applied to characterize the performance of the mmWave-NOMA transmission scheme, by using the key features of mmWave networks, i.e., mmWave transmission is highly directional and potential blockages will thin the user distribution. Two beamforming approaches that can effectively reduce  feedback have also been proposed to the addressed mmWave-NOMA communication networks, and the performance for the scenario with multiple beams has also been studied.  The provided simulation results have demonstrated that the developed analytical results are accurate, and the proposed  mmWave-NOMA transmission schemes yield significant performance gains over conventional mmWave-OMA schemes.
\vspace{-0.5em}
 %\linespread{1.2}
 \bibliographystyle{IEEEtran}
\bibliography{IEEEfull,trasfer}
 
  \end{document}